\documentclass[reqno]{amsart}

\usepackage{mathrsfs}
\usepackage{dsfont}
\usepackage[hypertex]{hyperref}  

\newcommand{\R}[0]{\mathds{R}}
\newcommand{\C}[0]{\mathds{C}}
\newcommand{\N}[0]{\mathds{N}}
\newcommand{\Z}[0]{\mathds{Z}}

\newcommand{\T}[0]{\mathds{T}}

\newtheorem{theorem}{Theorem}[section]
\newtheorem{corollary}[theorem]{Corollary}
\newtheorem{lemma}[theorem]{Lemma}
\newtheorem{proposition}[theorem]{Proposition}
\newtheorem{assumption}{Assumption}
\newtheorem{condition}{Condition}
\newtheorem{remark}[theorem]{Remark}
\newtheorem{definition}[theorem]{Definition}

\newcommand{\eref}[1]{(\ref{#1})}

\newcommand{\ol}[1]{\overline{#1}}

\newcommand{\ep}[0]{\varepsilon}
\newcommand{\osc}[0]{\mathrm{osc}}

\numberwithin{equation}{section}

\begin{document}

\title[complex geometric optics for symmetric hyperbolic systems II]{complex geometric optics for symmetric hyperbolic systems II: nonlinear theory in one space dimension}

\author{Omar Maj}

\address{Max-Planck-Institut f\"ur Plasmaphysik, D-85748 Garching, Germany}

\email{omaj@ipp.mpg.de}

\begin{abstract}
This is the second part of a work aimed to study complex-phase oscillatory solutions of nonlinear symmetric hyperbolic systems. We consider, in particular, the case of one space dimension. That is a remarkable case, since one can always satisfy the \emph{naive} coherence condition on the complex phases, which is required in the construction of the approximate solution. Formally the theory applies also in several space dimensions, but the \emph{naive} coherence condition appears to be too restrictive; the identification of the optimal coherence condition is still an open problem.
\end{abstract}


\maketitle

\section{Introduction}
\label{intro}

This paper constitutes the second part of a work dedicated to the analysis of oscillatory waves with complex phases, the theory of which is usually referred to as complex geometric optics in the applied literature. More specifically, in the first part (\cite{MI}, referred to as Part I throughout the paper), complex geometric optics for linear symmetric hyperbolic systems have been put forward as a preparatory study for nonlinear systems; now we move to the case of quasilinear first-order systems in a single spatial dimension. The general problem in several space dimensions is still open. The reason is that a coherence hypothesis on the complex phases is needed in order to control the resonance of waves, but the \emph{naive} condition, obtained naturally from the formal analysis, appears to be too strong; the identification of the optimal coherence hypothesis is closely connected to the formulation of the appropriate class of profiles, cf., section \ref{extension}.

When resonances occur, coherence conditions are crucial even for real-phase oscillatory waves, i.e., in the standard nonlinear geometric optics. Resonances have been clearly identified by Hunter and Keller \cite{HK} and the theory for resonantly interacting waves with real phases has been achieved by Majda and Rosales in one space dimension, \cite{MR}, and by Hunter, Majda and Rosales for several space dimensions, \cite{HMR}. An important result of such works is that, in several space dimensions, a strong coherence hypothesis on the phases is required in order to avoid hidden focusing, i.e., focusing of phases generated through resonant interactions; then, these results have been further developed and refined in the rigorous analysis of Joly, M\'etivier and Rauch, both in one \cite{JMRone} and several space dimensions \cite{JMR1}. A careful analysis of the coherence and focusing, including a number of clarifying examples, is also given by Joly, M\'etivier and Rauch \cite{JMR2}. An account of the main developments of the theory can be found in the lectures by Rauch \cite{R} and in the review by Dumas \cite{D}. 

In our case, the main simplification in one single spatial dimension is that no diffraction effects are present in the following sense. In complex geometric optics, diffraction is described as a coupling between the equations for the real and the imaginary part of the complex phase; in one space dimension such equations are decoupled. This allows us to make use of the \emph{naive} coherence condition which, in addition, turns out to be always satisfied. 

In several space dimensions, diffraction sets in; when the \emph{naive} coherence hypothesis is satisfied, the theory developed here is still applicable (after a straightforward generalization). However, one can readily recognize that such a hypothesis is violated even in very simple cases.

\section{Assumptions and Main Results}
\label{results}

Let us consider the quasilinear system,
\begin{equation}
\label{new:L}
L(t,x,u,\partial u) = \partial_t u + A(t,x,u)\partial_x u + F(t,x,u)=0,
\end{equation}
for $u(t,x) \in \C^N$ and $(t,x) \in  \R^2$. Here, $A(t,x,u)$ and $F(t,x,u)$ are smooth functions of $(t,x,u)$ with values in the space $\mathrm{End}(\C^N)$ of $N\times N$ complex matrices and in $\C^N$, respectively; particularly, smoothness with respect to the complex variables $u$ means that the functions are real-differentiable with respect to $(u,\ol{u})$, with the dependence on the complex conjugate $\ol{u}$ being always implied. We shall consider solutions in a bounded open set $\Omega \subset \R^2$ taking values in an compact set $K \subset \C^N$.

To some extent, we shall address the semilinear case separately; we recall that \eref{new:L} is called semilinear when the matrix $A$ is independent on the unknown $u$.

\begin{assumption}
\label{new:zero}
If the system \eref{new:L} is semilinear, $F(t,x,u)$ is a polynomial in $(u,\ol{u})$ with $F(t,x,0)=0$.
\end{assumption}

\begin{assumption}
\label{new:hyperb}
The system \eref{new:L} is symmetric and strictly hyperbolic in $\ol{\Omega}\times K$, that is, the matrix $A(t,x,u)$ is Hermitian in a neighbourhood of $\ol{\Omega}\times K$ and its eigenvalues $\lambda_l(t,x,u)$ are distinct with a uniform bound on the distance.
\end{assumption}

In the quasilinear case we consider solutions that are small perturbation of a known exact solution $u_0 \in C^\infty(\ol{\Omega};K)$ and we define the matrix $A_0(t,x) = A(t,x,u_0(t,x))$. Then, the principal symbol of the differential operator \eref{new:L} linearized around $u=u_0$ is
\begin{equation}
\label{new:symbol}
\sigma_{L_0}(t,x,\tau,\xi) = i(\tau + A_0(t,x)\xi), \quad (\tau,\xi) \in \R^2\setminus \{0\},
\end{equation}
where $L_0(t,x,\partial) =\partial_t + A_0(t,x)\partial_x$ denotes the principal part of the linearized operator; the same principal part is found for the semilinear case without linearization and with $A_0(t,x) = A(t,x)$.

Over $\ol{\Omega}$, the characteristic variety of $L_0$, i.e., the locus of $\det\sigma_{L_0}(t,x,\tau,\xi)=0$, amounts to a set of smooth disjoint submanifolds of $T^*\R^2\setminus 0$ (the cotangent bundle without the zero section) given by
\begin{equation*}
f_l(t,x,\tau,\xi) = \tau + \lambda_{l,0}(t,x)\xi = 0,
\end{equation*}
where $f_l$ are the distinct eigenvalues of $\sigma_{L_0}$ and $\lambda_{l,0}(t,x) =\lambda_{l}(t,x,u_0(t,x))$. Moreover, we note that the eigenspaces of $\sigma_{L_0}$ are exactly the same as the eigenspaces of $A_{0}$, thus, in particular, they do not depend on $\xi$. Let us denote by $\pi_l(t,x)$ be the projector on the $l$-th eigenspace.

As in the linear complex geometric optics \cite{MI}, for $T>0$ and $\underline{x} \in\R$, we shall address the Cauchy problem for \eref{new:L} on the domain of determinacy
\begin{equation*}
\ol{\Omega}=\{(t,x) \in \R^2; 0\leq t\leq T, |x-\underline{x}| \leq \rho - ct\},
\end{equation*}
where $\rho >cT$ and $c$ is the finite propagation speed for the system \eref{new:L} so that
\begin{equation*}
cI + A(t,x,u)x/|x| \geq 0, \quad (t,x,u)\in \ol{\Omega} \times K,
\end{equation*}
$I$ being the identity matrix. We shall denote $\ol{X^{t'}} = \ol{\Omega}\cap\{t=t'\}$ the space-like slices and by $X^t$ the interior of the closed interval $\ol{X^t}$. Data are given at $t=0$ in the form
\begin{equation}
\label{new:datum-quasi}
u^\ep_{|t=0}(x) = u_0(0,x) + \ep h^\ep(x), \qquad x \in \ol{X^o},
\end{equation}
in the quasilinear case and
\begin{equation}
\label{new:datum-semi}
u^\ep_{|t=0}(x) = h^\ep(x), \qquad x \in \ol{X^o},
\end{equation}
in the semilinear case with
\begin{equation*}
h^\ep(x) = \sum_{\mu=1}^m h_\mu(x) e^{i\psi_\mu(x)/\ep}, \quad \ep \in \R_+.
\end{equation*}
Here, $\psi = (\psi_\mu) \in C^\infty(\ol{X^o};\C^m)$ is a complex phase in $\C^m$ and $h_\mu \in C^\infty(\ol{X^o};\C^N)$ are the amplitudes. As in the linear theory we have the following definitions and assumptions. 

\begin{definition}[Complex phase]
\label{complex_phase}
For $m \geq 0$ integer, $\phi \in C^\infty(\ol{\mathcal{O}}; \C^m)$ in a bounded open set $\mathcal{O} \subset \R^n$ is a complex phase iff, for each component $\phi_\mu = \varphi_\mu +  \chi_\mu$, one has $d\varphi_\mu \not =0$ and $\chi_\mu \geq 0$ in $\ol{\mathcal{O}}$.
\end{definition}

\begin{assumption}
\label{Rzero}
The locus of $\mathrm{Im} \psi_{\mu}(x)=0$ amounts to the set of points $x^o_{\mu,\ell} \in X^o$, $\ell = 1,2,\ldots$, $x^o_{\mu,\ell} = x^o_{\mu',\ell'}$ iff $(\mu,\ell) = (\mu',\ell')$ and $d^2\mathrm{Im}\psi_\mu(x^o_{\mu,\ell}) >0$.
\end{assumption}

After splitting the amplitudes as appropriate, we can always write the initial datum so that the following condition is satisfied.

\begin{condition}
\label{pol}
For every $\mu$ there is $l=l(\mu)$ such that $\pi_l(0,x)h_\mu(x) = h_\mu(x)$.
\end{condition}

The standard nonlinear geometric optics \cite{R} applies when $\mathrm{Im}\psi_\mu \equiv 0$; then, one looks for asymptotic solutions to \eref{new:L} in the form
\begin{equation*}
u^\ep(t,x) \sim u_0(t,x) + \sum_{j=1}^{+\infty} \ep^j U_j\big(t,x,\varphi(t,x)/\ep\big), \quad \ep \to 0,
\end{equation*}
where $u_0 \in C^\infty(\ol{\Omega}; K)$ and $U_j \in C^\infty(\ol{\Omega} \times \R^m;\C^N)$ are $2\pi$-periodic in the last variable, i.e., $U_j(t,x,\theta + 2\pi g) = U_j(t,x,\theta)$ for $g\in\Z^m$, whereas $\varphi \in C^\infty(\ol{\Omega};\R^m)$ is the vector of the $m$ real-valued phases $\varphi_\mu (t,x)$, with $d\varphi_\mu \not =0$. In the semilinear case, one can set $u_0=0$ and let the sum start from $j=0$, thus, getting a fully nonlinear geometric optics solution. The functions $U_j(t,x,\theta)$, in particular, are called \emph{profiles}, and, in view of the periodicity, they are better defined as functions of class $C^\infty(\ol{\Omega}\times \T^m;\C^N)$ where $\T = \R/2\pi\Z$ is the torus; on representing such profiles by means of the Fourier series, one can see that they are generalization of the purely exponential profiles of the linear geometric optics, 
\begin{equation*}
u^\ep(t,x) = \sum a^\ep_\mu(t,x) e^{i\varphi_\mu(t,x)/\ep},
\end{equation*} 
with the main nonlinear effects accounted for, namely,
\begin{itemize}
\item[1.] \emph{generation of harmonics}: in addition to the fundamental harmonic $e^{i\theta_\mu}$, all the other harmonics $e^{ig_\mu\theta_\mu}$, $g_\mu \in\Z$ are accounted for; 
\item[2.] \emph{resonance of phases}: the $m$ phases $\varphi_\mu$ are dealt with all together so that the interaction among them can occur through terms of the form $\langle g,\varphi \rangle = \sum_\mu g_\mu \varphi_\mu$.
\end{itemize}

Analogously, we have to address the proper nonlinear generalization of the complex geometric optics ansatz, \cite[and references therein]{MI}, 
\begin{equation*}
u^\ep(t,x) = \sum a^\ep_\mu(t,x) e^{i\phi_\mu(t,x)/\ep}.
\end{equation*}
Here, $\phi=(\phi_\mu)_\mu$ is a vector of complex-valued phases, cf., definition \ref{complex_phase}. 

In section \ref{extension}, we shall define the space $PC_\osc^\infty(\ol{\Omega};\C^N)$ of oscillatory complex geometric optics profiles $\mathcal{U}(t,x,z)$; roughly speaking those are smooth functions of $(t,x) \in \ol{\Omega}$ and $z = \theta + i r \in \C^m$, with $\theta \in \T^m$ and $r \in \ol{\R}_+^m$, that can be written as a superposition of harmonics of the form $e^{i\langle g,\theta \rangle - \langle \gamma,r \rangle}$ for $(g,\gamma) \in \Z^m \times \N^m$ satisfying the conditions $g\not=0$ (oscillatory profiles) and $|g_\mu|\leq \gamma_\mu$.

For $p \geq 0$, we search for approximate solutions in the form
\begin{equation}
\label{new:ansatz}
v^\ep(t,x) = u_0(t,x) + \ep^p \mathcal{U}^\ep (t,x, \phi/\ep), \quad \mathcal{U}^\ep = \mathcal{U}^{(0)} + \ep \mathcal{U}^{(1)},
\end{equation}
where $\mathcal{U}^{(i)} \in PC_\osc^\infty(\ol{\Omega};\C^N)$, $i=1,2$, and where $\phi \in C^\infty(\ol{\Omega};\C^m)$ is a complex phase. We put $p=1$ in the quasilinear case and $u_0=0$, $p=0$ in the semilinear case.

We shall see in section \ref{extension} that the space of oscillatory profiles $PC^\infty_\osc$ is a subspace of the algebra $PC^\infty$ of generic complex geometric optics profiles, but it is not a subalgebra: oscillatory profiles exclude all the cases in which rectification effects \cite{Rect} are present, that is, when the nonlinear interaction of two oscillatory waves gives rise to a non-oscillatory wave that sums up to the background field $u_0$. This leads to a technical assumption on the nonlinear operators derived in the next section \ref{sec:formal-expansion}. 

\subsection{Formal expansion and rectification}
\label{sec:formal-expansion}

Upon substituting \eref{new:ansatz} into \eref{new:L}, straightforward calculations show that
\begin{equation}
\label{new:formal-expansion}
\begin{aligned}
L(t,x,v^\ep,\partial v^\ep) &= \ep^{p-1} P(t,x,\partial_z,\partial_{\ol{z}}) \mathcal{U}^{(0)}(t,x,\phi/\ep) \\
& \quad + \ep^p \big[P(t,x,\partial_z,\partial_{\ol{z}}) \mathcal{U}^{(1)}(t,x,\phi/\ep)
 + N(\mathcal{U}^{(0)})(t,x,\phi/\ep)\big] \\
& \quad + O(\ep^{p+1}),
\end{aligned}
\end{equation}
for $\ep \to 0$. Here, $p=1$ and $N(\mathcal{U})=\mathscr{B}_0(\mathcal{U})$ in the quasilinear case whereas $p=0$ and $N(\mathcal{U}) = L(t,x,\mathcal{U},\partial\mathcal{U})$ in the semilinear case. The foregoing expansion is obtained by the same formal calculations as in \cite{JMR1} with only two differences: we restrict our analysis to finite order accuracy and we find the derivatives $(\partial_z, \partial_{\ol{z}})$ instead of $\partial_\theta$. The linear and nonlinear operators occurring in equation \eref{new:formal-expansion} are
\begin{equation*}
P(t,x,\partial_z,\partial_{\ol{z}}) = \sum_{\mu=1}^m \Big[\big( \partial_t \phi_\mu + A_0(t,x)\partial_x\phi_\mu \big)\partial_{z_\mu} + \big( \partial_t \ol{\phi}_\mu  + A_0(t,x)\partial_x\ol{\phi}_\mu \big)\partial_{\ol{z}_\mu} \Big],
\end{equation*}
and
\begin{align*}
\mathscr{B}_0(\mathcal{U}) &= L_0(t,x,\partial) \mathcal{U} + \sum_{\mu=1}^m \Big[ \Big( \partial_u A(t,x,u_0) \mathcal{U} \frac{\partial \phi_\mu}{\partial x}\Big) \partial_{z_\mu}  \mathcal{U} \\
&\ + \Big( \partial_u A(t,x,u_0) \mathcal{U} \frac{\partial \ol{\phi}_\mu}{\partial x}\Big) \partial_{\ol{z}_\mu}  \mathcal{U} + \Big( \partial_{\ol{u}} A(t,x,u_0) \ol{\mathcal{U}} \frac{\partial \ol{\phi}_\mu}{\partial x}\Big) \partial_{\ol{z}_\mu}  \mathcal{U} \\ 
&\ +\Big( \partial_{\ol{u}} A(t,x,u_0) \ol{\mathcal{U}} \frac{\partial \phi_\mu}{\partial x}\Big) \partial_{z_\mu}  \mathcal{U} \Big] + \Big( \partial_u F(t,x,u_0)\mathcal{U} + \partial_{\ol{u}} F(t,x,u_0)\ol{\mathcal{U}} \Big)\\
&\ + \Big(\partial_u A(t,x,u_0)\mathcal{U} + \partial_{\ol{u}} A(t,x,u_0) \ol{\mathcal{U}} \Big)\partial_{x} u_0.
\end{align*}
We note that $\mathscr{B}_0$ exhibits a non-linearity of Burgers type due to the differentials $(\partial_u A_j,\partial_{\ol{u}}A_j)$. The nonlinear terms in $N(\mathcal{U})$ are polynomials in $\mathcal{U},\ol{\mathcal{U}}$ (in the semilinear case this follows from assumption \ref{new:zero}), thence $N(\mathcal{U})$ maps the algebra $PC^\infty$ into itself; we assume that the nonlinear terms are such that rectification does not occur, i.e., the oscillatory character of profiles is preserved.

\begin{assumption}
\label{rect}
The nonlinear operator $N:PC^\infty \to PC^\infty$ restricts to an operator $:PC^\infty_\osc \to PC^\infty_\osc$ still denoted by $N$.
\end{assumption}

In the present context, rectification behaves exactly in the same way as for periodic profiles of standard geometric optics; we refer to the paper by Joly, M\'etivier and Rauch \cite{Rect} for further details and examples.

\subsection{Complex phases}
\label{subsec:phases}

The complex phase $\phi(t,x)$ is obtained as follows. Let us consider the smooth vector fields,
\begin{equation*}
V_l(t,x) = \partial_t + \lambda_{l,0}(t,x)\partial_x, \quad V_\mu(t,x)=V_{l(\mu)}(t,x),
\end{equation*}
where $l(\mu)$ is defined in condition \ref{pol}. The integral lines of $V_l$ are just the standard geometric optics rays.

\begin{assumption}
\label{new:ref-man}
For every $\mu \in \{1,\ldots,m\}$, the integral line of $V_\mu$ passing through $x^o_{\mu,\ell}$ is defined in a neighbourhood of $[0,T]$ and crosses transversally at $X^T$ the lateral boundary $\partial \Omega \setminus \ol{X^o}$. We denote by $R_{\mu,\ell}$ the intersection of the integral line with $\ol{\Omega}$, then $\{R_{\mu,\ell}\}_{\mu,\ell}$ is a family of disjoint one-dimensional submanifolds of $\ol{\Omega}$.
\end{assumption}

It is worth noting that integral curves of $V_\mu$ are always transversal to the space-like sections $X^t$, hence, $R_{\mu,\ell}$ is of the form $(t,x_{\mu,\ell}(t))$. Upon defining $s=x-x_{\mu,\ell}(t)$ we obtain a coordinate patch $\kappa : \mathcal{O}_{\mu,\ell} \ni (t,x) \mapsto (t,s) \in [0,T] \times \mathcal{I}_s$, with $\mathcal{O}_{\mu,\ell}$ a relatively open neighbourhood of $R_{\mu,\ell}$ in $\ol{\Omega}$ and $\mathcal{I}_s$ an open interval in $\R$ containing zero. Such coordinates have the submanifold property for $R_{\mu,\ell}$, i.e., $R_{\mu,\ell} \cap \mathcal{O}_{\mu,\ell}$ is mapped into $\{s=0\}$, and, since the curves $R_{\mu,\ell}$ do not cross each other, we can take the sets $\mathcal{O}_{\mu,\ell}$ pairwise disjoint. Let us define the complex phase $\phi_\mu$ in each relatively open $\mathcal{O}_{\mu,\ell}$ by
\begin{equation}
\label{new:eikonal}
\phi_\mu(t,x) = \varphi_{\mu,\ell}(t) + \xi_{\mu,\ell}(t)s + \Phi_{\mu,\ell}(t)s^2/2, \quad s=x-x_{\mu,\ell}(t),
\end{equation}
where $\varphi_{\mu,\ell},\xi_{\mu,\ell} \in C^\infty([0,T];\R)$ and $\Phi \in C^\infty([0,T];\C)$ are unknown functions, cf., section 4 of Part I. We set $\varphi_0(t) = \psi(x^o_{\mu,\ell})$ which is real because of assumption \ref{Rzero}. The remaining unknowns are determined by the system of linear ordinary differential equations
\begin{equation*}
\left\{
\begin{aligned}
& \xi_{\mu,\ell}'(t) +\alpha(t)\xi_{\mu,\ell}(t)=0,\\
& \Phi'_{\mu,\ell}(t) +2 \alpha(t)\Phi_{\mu,\ell}(t) + \beta(t)\xi_{\mu,\ell}(t)=0,
\end{aligned}\right.
\end{equation*}
with initial values
\begin{equation*}
\xi_{\mu,\ell}(0) = d \psi_\mu(x^o_{\mu,\ell}), \quad \Phi_{\mu,\ell}(0) = d^2 \psi_{\mu}(x^o_{\mu,\ell}).
\end{equation*}
where $d\psi_\mu(x^o_{\mu,\ell})$ is real and $d^2\mathrm{Im}\psi_\mu(x^o_{\mu,\ell}) >0$. The coefficients are given by
\begin{equation*}
\alpha(t)= \partial \lambda_{l(\mu),0}/\partial x|_{R_{\mu,\ell}},\quad \beta(t)= \partial^2\lambda_{l(\mu),0}/\partial x^2|_{R_{\mu,\ell}}.
\end{equation*}
The foregoing construction of the complex phase $\phi$ is a special case of the more general procedure addressed in section 4 of Part I, cf., in particular, remark 4.2 of Part I. Here, it is worth noting that the nonlinear coupling between the real and imaginary parts of the phases is no longer present: this entails the fact that in one space dimension diffraction does not exists. 

The solution is global on $[0,T]$ and one can see that $\mathrm{Im} \Phi_{\mu,\ell}(t) >0$ in $[0,T]$, cf., also, remark 4.3 of Part I. The following result can be proved either directly or by proposition 4.4 of Part I.

\begin{proposition}
\label{new:riccati}
Let assumptions \ref{Rzero} and \ref{new:ref-man} be satisfied together with condition \ref{pol} and let $\phi_\mu$, $\mu \in \{1,\ldots,m\}$, be complex phases satisfying \eref{new:eikonal} in each $\mathcal{O}_{\mu,\ell}$. Then,
\begin{equation*}
V_\mu(t,x)\phi_\mu(t,x) = O(|s|^3), \qquad \phi_{\mu|t=0}(x)-\psi_\mu(x) = O(|s|^3),
\end{equation*}
in $\mathcal{O}_{\mu,\ell}$ and $\mathcal{O}^o_{\mu,\ell}=\mathcal{O}_{\mu,\ell}\cap X^o$, respectively.
\end{proposition}

However, we see that the expansion \eref{new:eikonal} does not determine the phases globally on $\ol{\Omega}$. Hence, we define an equivalence relation in $C^\infty(\ol{\Omega};\C^m)$, cf., also Part I,
\begin{equation}
\label{equivalence}
\begin{aligned}
&\text{$\phi,\phi' \in C^\infty(\ol{\Omega};\C^m)$ are equivalent iff the components $\phi_\mu$, $\phi'_\mu$ have the}\\ 
&\text{same Taylor polynomial of degree $k$ in the variable $s$ near $R_{\mu,\ell}$ for all $\ell$}.
\end{aligned}
\end{equation}
Then, \eref{new:eikonal} characterizes an unique equivalence class (with $k=2$) and any representative element $\phi$ satisfies $\mathrm{Im}\phi_\mu=0$ on $R_{\mu,\ell}$ for every $\ell$; we pick $\phi$ so that $\mathrm{Im}\phi_\mu = 0$ only on $\bigcup_\ell R_{\mu,\ell}$. The union $R_\mu=\bigcup_\ell R_{\mu,\ell}$ is a closed one-dimensional submanifold in $\ol{\Omega}$ called reference manifold for $\phi_\mu$; analogously, $R = \bigcup_\mu R_\mu$ is the reference manifold for $\phi$.

\subsection{Profiles}
\label{subsec:profiles}

The oscillatory profiles of complex geometric optics are superposition of harmonics $e^{i\langle g,\theta \rangle - \langle \gamma,r \rangle}$ with $(g,\gamma) \in \Sigma_\osc$ and
\begin{equation}
\label{new:spectrum-osc}
\Sigma_\osc = \{(g,\gamma) \in \dot{\Z}^m \times \dot{\N}^m;\ g=(g_\mu),\ \gamma=(\gamma_\mu),\  |g_\mu| \leq \gamma_\mu\},
\end{equation}
is their spectrum, $\dot{\Z}^m = \Z^m\setminus\{0\}$ and $\dot{\N}^m = \N^m \setminus\{0\}$. This very specific form of the spectrum is motivated in section \ref{extension}.

Let $\phi(t,x)$ be a representative for the equivalence class characterized by \eref{new:eikonal}, and let us define the following sets,
\begin{align}
\label{rset}
&\mathscr{R}^\phi = \{(t,x,g,\gamma) \in \ol{\Omega} \times \Sigma_\osc ;\ \mathrm{Im}\Psi(g,\gamma;\phi)=0 \},\\
\label{cset}
&\mathscr{C}^{\phi} = \{(t,x,g,\gamma) \in \mathscr{R}^\phi ;\ \det \sigma_{L_0}(t,x,d\Psi(g,\gamma;\phi)) = 0\},
\end{align}
where $\Psi(g,\gamma;\phi)=\langle g,\varphi \rangle + i\langle \gamma,\chi \rangle$, $\varphi = \mathrm{Re}\phi$ and $\chi = \mathrm{Im}\phi$. We have, in particular, $\mathrm{Im}\Psi(g,\gamma;\phi) = \langle \gamma,\chi \rangle = 0$; since $\gamma_\mu\in\N$ and $\chi_\mu$ vanishes only on the disjoint curves $R_{\mu,\ell}$, $\mathscr{R}^\phi$ is constant over each reference manifold $R_\mu$, namely,
\begin{equation*}
\mathscr{R}^\phi = \bigcup_{\mu =1}^m R_\mu \times \Sigma_\mu,
\end{equation*}
where
\begin{equation*}
\Sigma_\mu = \{(g,\gamma)\in \Sigma_\osc; \text{$\gamma=(\gamma_\nu)_\nu$ with $\gamma_\nu=0$ for $\nu \not=\mu$}\},
\end{equation*}
is independent on $\phi$; when $(g,\gamma) \in \Sigma_\mu$, the imaginary part of the complex phase $\Psi(g,\gamma;\phi)=g_\mu\varphi_\mu + i \gamma_\mu\chi_\mu$ vanishes on $R_\mu$ only. The crucial point is that, in one spatial dimension, also $\mathscr{C}^\phi$ is constant over each $R_\mu$ and this, roughly speaking, is the coherence property we need, cf., section \ref{intro}.

\begin{proposition}
\label{new:coherence1}
With assumption \ref{Rzero} and \ref{new:ref-man} and condition \ref{pol} satisfied, let $\phi \in C^\infty(\ol{\Omega};\C^m)$ be a complex phase satisfying \eref{new:eikonal} near $R$, then $\mathscr{C}^\phi = \mathscr{R}^\phi$.
\end{proposition}

\begin{proof}
By definition $\mathscr{C}^\phi \subseteq \mathscr{R}^\phi$. The converse can be proved on noting that, for $(t,x,g,\gamma)\in \mathscr{R}^\phi$, there is $\mu \in \{1,\ldots,m\}$ such that $(t,x) \in R_\mu$ and $\Psi(g,\gamma;\phi)=g_\mu\varphi_\mu + i \gamma_\mu\chi_\mu$ which solves the equation 
\begin{equation*}
V_\mu(t,x)\Psi(g,\gamma;\phi)=O(|s|^3),
\end{equation*}
near $R_\mu$. From the identity $\det \sigma_{L_0}(t,x,d\Psi) = \prod_l\big[V_l(t,x)\Psi(g,\gamma;\phi) \big]^{m_l}$, the product being over all the eigenvalues $\lambda_l$ with the corresponding multiplicity $m_l$, we have $\det \sigma_{L_0}(t,x,d\Psi)=0$ on $R_\mu$.
\end{proof}

\begin{proposition}
\label{new:coherence2}
Under the same hypothesis of proposition \ref{new:coherence1}, for every $(g,\gamma)\in\Sigma_\mu$ and $l \not =l(\mu)$ one has
\begin{equation*}
|V_l(t,x)\Psi(g,\gamma;\phi)| \geq C > 0,
\end{equation*}
near the one-dimensional submanifold $R_\mu$.
\end{proposition}

\begin{proof}
If there would be an $l$ such that, for every $C>0$, $|V_l\Psi(g,\gamma;\phi)| \leq C$ in a point on $R_\mu$, then, in that point, $|\lambda_{l,0}-\lambda_{l(\mu),0}| |\partial_x\varphi_\mu| \leq C$ against the hypothesis of strict hyperbolicity for which $\lambda_l$ should be distinct with a uniform bound on the distance.
\end{proof}

Let us now define the cut-off functions $\omega_{\mu,\ell}(t,x)$ as in section 6 of Part I: we pick $\omega_{\mu,\ell}^o \in C^\infty_0(\mathcal{O}_{\mu,\ell}^o)$, $\mathcal{O}_{\mu,\ell}^o = \mathcal{O}_{\mu,\ell} \cap X^o$, such that $\omega^o_{\mu,\ell} \equiv 1$ in a neighbourhood of the point $x^o_{\mu,\ell}$ and we make use of the coordinates $(t,s)$ in order to extend $\omega^o_{\mu,\ell}$ to the relatively open $\mathcal{O}_{\mu,\ell} \subseteq \ol{\Omega}$, namely, we set $\omega_{\mu,\ell}(t,s)=\omega^o_{\mu,\ell}(s)$.

By using the series representation for oscillatory profiles, cf., section \ref{extension},
\begin{equation*}
\mathcal{U}(t,x,z) = \sum_{(g,\gamma)\in \Sigma_\osc}\widehat{U}(t,x,g,\gamma)e^{i\Psi(g,\gamma;z)},
\end{equation*}
and the functions $\omega_{\mu,\ell}(t,x)$, we define the operator
\begin{equation}
\label{new:E}
\mathds{E} \mathcal{U}(t,x,z) = \sum_{(g,\gamma) \in \Sigma_\osc} \pi(t,x,g,\gamma) \hat{U}(t,x,g,\gamma)e^{i\Psi(g,\gamma;z)}, 
\end{equation}
where $\pi(t,x,g,\gamma) = \sum_\ell \omega_{\mu,\ell}(t,x)\pi_{l(\mu)}(t,x)$ when $(g,\gamma)\in\Sigma_\mu$, otherwise we set arbitrarily $\pi(t,x,g,\gamma)=I$. Analogously, we set
\begin{equation}
\label{new:Q}
Q \mathcal{U}(t,x,z) = -i\sum_{(g,\gamma) \in \Sigma_\osc} \mathscr{Q}^\phi(t,x,g,\gamma) \hat{U}(t,x,g,\gamma)e^{i\Psi(g,\gamma;z)}, 
\end{equation}
where $\mathscr{Q}^\phi(t,x,g,\gamma)$ is a smooth extension (e.g., obtained by using $\omega_{\mu,\ell}$) to a compact neighbourhood of $R_\mu$ of 
\begin{equation*}
\sum_{l\not=l(\mu)} (V_{l}(t,x)\Psi(g,\gamma;\phi) \rangle)^{-1}\pi_{l'}(t,x),\quad (t,x)\in R_\mu,
\end{equation*}
when $(g,\gamma)\in\Sigma_\mu$ and we set arbitrarily $\mathscr{Q}^\phi(t,x,g,\gamma)=0$ otherwise. The operator $Q$ is well defined in virtue of proposition \ref{new:coherence2}. Let us define the equivalence relation
\begin{equation}
\label{FM-equivalence}
\begin{aligned}
&\text{two Fourier multipliers $A_1,A_2:PC_\osc^\infty \to PC_\osc^\infty$ are equivalent iff}\\
&\text{for every $\mu$ their coefficients with $(g,\gamma) \in \Sigma_\mu$ have the same Taylor}\\ 
&\text{polynomial of degree $k\geq 0$ in the variable $s$ near $R_\mu$}.
\end{aligned}
\end{equation}
Then, we can take any other pair of Fourier multipliers that are equivalent to \eref{new:E} and \eref{new:Q} with $k=2$ and $k=0$, respectively, cf., also propositions \ref{smooth-EQ} and \ref{modulo}.

The profiles in \eref{new:ansatz} are given iteratively by
\begin{equation*}
\mathcal{U}^{(0)}(t,x,z) = \mathds{E} \underline{\mathcal{U}} (t,x,z), \qquad \mathcal{U}^{(1)}(t,x,z) = - Q N(\mathcal{U}^{(0)})(t,x,z),
\end{equation*}
where $\underline{\mathcal{U}} \in PC_\osc^\infty(\ol{\Omega};\C^N)$ is a smooth extension to a compact neighbourhood of $R$ in $\ol{\Omega}$ of the solution $\mathcal{U} \in PC^\infty_\osc(R;\C^N)$ of the following Cauchy problem. First, let us write $z=\theta +ir$ and define
\begin{align*}
&B(t,x,\mathcal{U})\partial_\theta \mathcal{U} = \sum_{\mu=1}^m  \big[\partial_x \varphi_\mu \partial_u A \mathcal{U} + \partial_x \varphi_\mu \partial_{\ol{u}} A \ol{\mathcal{U}} \big]\partial_{\theta_\mu} \mathcal{U},\\
&C(t,x,\mathcal{U})\mathcal{U} = (\partial_u A \mathcal{U} + \partial_{\ol{u}} A \ol{\mathcal{U}}) \partial_x u_0 + (\partial_{u}F \mathcal{U} + \partial_{\ol{u}} F \ol{\mathcal{U}} ),
\end{align*}
in the quasilinear case and $B=0$, $C(t,x,\mathcal{U})\mathcal{U} = F(t,x,\mathcal{U})$ in the semilinear case; the latter definition of $C$ is consistent since $F(t,x,0)=0$ in view of assumption \ref{new:zero}. Finally, let $\mathds{E}_0$ be the restriction of $\mathds{E}$ to the reference manifold $R$. Then, $\mathcal{U} \in PC^\infty_\osc(R;\C^N)$ is determined by the transport equation on $R$
\begin{equation}
\label{te}
\left\{
\begin{aligned}
&(I-\mathds{E}_0) \mathcal{U} = 0,\\
&\mathds{E}_0\big[L_0 + B(t,x,\mathcal{U})\partial_\theta + C(t,x,\mathcal{U})\big] \mathcal{U} =0,\qquad (t,x) \in R,\\
&\mathcal{U}_{|t=0}(x,z) = \mathcal{H}(x,z), \qquad x \in R^o=R \cap X^o,
\end{aligned}\right.
\end{equation}
where the initial datum is defined by the condition
\begin{equation*}
h^\ep (x) = \mathcal{H}\big(x,\psi(x)/\ep\big).
\end{equation*}
The profile $\mathcal{H} \in PC_\osc^\infty(\ol{X^o};\C^N)$ is obtained by defining its Fourier coefficients $\widehat{H}(t,x,g,\gamma) = h_\mu(x)$, for $g = \gamma =(0,\ldots,0,1,0,\ldots 0)$, with the unit in the $\mu$-th entry, and $\widehat{H} = 0$ otherwise. Condition \ref{pol} implies $(I-\mathds{E}_0) \mathcal{H}_{|R^o} = 0$, thus, $\mathcal{H}$ is admissible as initial condition for equation \eref{te}. The well-posedness of \eref{te} is addressed in section \ref{existence}.

\subsection{Main result}
\label{subsec:main} 

With the equivalence classes of complex phases and profiles, the ansatz \eref{new:ansatz} yields an equivalence class of approximate solutions. We note that, from one hand, the equivalence class of phases is readily determined by solving a set of ordinary differential equations; on the other hand, the existence of profiles is subordinated to the well-posedness of the Cauchy problem \eref{te} which is proved in section \ref{existence}.

\begin{proposition}
\label{main}
Let assumptions \ref{new:zero}-\ref{new:ref-man} be satisfied and let us write the initial datum so that condition \ref{pol} holds true. Then, there exists an equivalence class of functions $v^\ep \in C^\infty(\ol{\Omega};\C^N)$ such that, for $\ep \in (0,\ep_0]$, $0<\ep_0<1$,
\begin{itemize}
\item[a)] $|u_{|t=0}^\ep - v^\ep_{|t=0}| \leq C_1 \ep^{p+\frac{1}{2}}$ uniformly in $\ol{X^o}$;
\item[a)] $|L(t,x, v^\ep,\partial v^\ep) | \leq C_2 \ep^{p+\frac{1}{2}}$ uniformly in $\ol{\Omega}$.
\end{itemize}
Here, $p=1$ in the quasilinear case and $p=0$ in the semilinear case.
\end{proposition}

As for the existence of the exact solutions $u^\ep$, the pointwise argument used in Part I for the linear theory fails for nonlinear equations as the lifespan of each solution depends on $\ep$. In order to obtain a family of exact solutions $\{u^\ep\}_\ep$ bounded in $L^\infty$, we need to control higher order derivatives so that we can apply the Sobolev's embedding theorem; on the other hand, applying $\partial_s^k$ with $k$ large enough leads to terms of order $\ep^{-h}$, $h >0$, that are unbounded when $\ep \in (0,\ep_0]$. Therefore, it is natural to work with conormal distributions \cite[and references therein]{MeRi,Me} as in the approach of Alterman and Rauch to short pulses \cite{AR1,AR2}. The development of this ideas is left for future work. 

\section{Complex Geometric Optics Profiles}
\label{extension}

In this section, $\ol{\Omega}$ denotes a generic smooth complex manifold with boundary $\partial \Omega$; we shall define complex geometric optics profiles over $\ol{\Omega}$, that is, functions $\mathcal{U}(y,z)$ of $y \in \ol{\Omega}$ and $z = (z_\mu) \in \C^m$ with $\mathrm{Im} z_\mu \geq 0$ and with the following properties:
\begin{itemize}
\item[(i)] smoothness with respect to all the variables;
\item[(ii)] periodicity in $\mathrm{Re} z$, that is, $\mathcal{U}(t,x, z) =\mathcal{U}(t,x,z + 2\pi g)$, for any $g\in \Z^m$; 
\item[(iii)] closure of their space with respect to nonlinear partial differential operators with polynomial nonlinearity and coefficients in $C^\infty(\ol{\Omega})$.
\end{itemize}

Let us recall that, in standard geometric optics, periodic profiles are functions that belong to $C^\infty(\ol{\Omega} \times \T^m; \C^N)$. We shall replace the torus $\T$ with $\T_c = \ol{\C}_+/2\pi\Z$, where the action of the group $2\pi \Z$ on the ``upper-half'' complex plane $\ol{\C}_+ = \{ w \in\C ; \mathrm{Im} w \geq 0 \}$ is $w \mapsto w + 2\pi n$, $w \in \ol{\C}_+$ and $n \in \Z$. We see that $\T_c^m \cong  \T^m \times \ol{\R}^m_+$ since any $z \in \T_c^m$ can be written in the form $z = \theta + i r$ with $\theta \in \T^m$ and $r \in \ol{\R}^m_+$; the torus $\T^m$ is identified with $\T^m\times\{r=0\}$. Now, a function $\mathcal{U}$ of class $C^\infty(\ol{\Omega} \times \T_c^m; \C^N)$ satisfies conditions (i) and (ii) automatically.

\begin{definition}
\label{complex-profile-defn}
The spectrum of a generic complex geometric optics profile is 
\begin{equation}
\label{new:spectrum}
\Sigma =\{(g,\gamma)\in\Z^m \times \N^m;\ g=(g_\mu),\ \gamma=(\gamma_\mu),\ |g_\mu|\leq \gamma_\mu\}.
\end{equation}
Then, the algebra $PC^\infty(\ol{\Omega})$ is the subspace of $C^\infty(\ol{\Omega} \times \T_c^m)$ of series
\begin{equation}
\label{complex-profile-series}
\mathcal{U}(y,z) = \sum \widehat{U}(y,g,\gamma) e^{i \Psi(g,\gamma;z)},\qquad (g,\gamma) \in \Sigma,
\end{equation}
where $\widehat{U}(\cdot,g,\gamma) \in C^\infty(\ol{\Omega})$ and $|X_1\cdots X_M \widehat{U}(y,g,\gamma)| \leq C_{\alpha} |(g,\gamma)|^{-k}$ for every set of smooth tangent fields $X_1,\ldots,X_M$ in $\ol{\Omega}$ and for every $k \in \N$.
\end{definition}

\begin{remark}
\label{new:tensor-algebra}
If not specified the target space of profiles $\mathcal{U}$ is the tensor algebra of $\C^N$ and $PC^\infty(\ol{\Omega})$ is an algebra with respect to the pointwise tensor product.
\end{remark}

\begin{remark}
\label{new:motivation}
One may define profiles with the spectrum being the whole $\Z^m\times\N^m$. The particular choice of the spectrum $\Sigma$ is motivated as follows. In the linear complex geometric optics one considers harmonics $u_{\pm}(z)=a e^{\pm i\theta - r}$; on starting from these functions, a polynomial nonlinearity generates a profile with the spectrum contained in $\Sigma$. The choice of such a minimal spectrum greatly simplifies the analysis of the characteristic set $\mathscr{C}^\phi$ and, thus, of the coherence for complex phases.
\end{remark}

For every function $\mathcal{U} \in PC^\infty$, $X_1\cdots X_M \partial_{z,\bar{z}}^\beta \mathcal{U}$ is bounded on $\ol{\Omega} \times \T_c^m$ and $PC^\infty$ is a Fr\'echet space with seminorms 
\begin{equation*}
\|\mathcal{U}\|_k =  \sum_{|\alpha| + |\beta| \leq k} \sup | (X_1,\ldots,X_n)^\alpha \partial_{z,\bar{z}}^\beta \mathcal{U}(y,z)|
\end{equation*}
with $X_1,X_2,\ldots,X_n$ being the generators of the Lie algebra of smooth vector fields on $\ol{\Omega}$: this is the $C^\infty$-topology on $\ol{\Omega}\times \T_c^m$. Particularly, that $PC^\infty$ is closed can be proved by using the following expression for the coefficients $\widehat{U}$. As a function of $(\theta = \mathrm{Re} z, r =\mathrm{Im} z)$, the profile $\mathcal{U}(y,z)=\mathcal{U}(y,\theta,r)$ can be extended to a function of $(y,\theta, w)$ with $w \in \C^m$ $r_\mu=\mathrm{Re}(w_\mu)  \geq 0$; then, we have
\begin{equation*}
\widehat{U}(y,g,\gamma) = \Big(\frac{-i}{4\pi^2}\Big)^m \int_{i\mathds{T}^m} \Big(\int_{\T^m} \mathcal{U}(y,\theta,w) e^{- i\langle g,\theta \rangle} d\theta\Big)e^{-\langle \gamma,w \rangle}dw,
\end{equation*}
where $i\mathds{T}^m = \{ w \in \C^m ; w = i \theta',\ \theta' \in \T^m\}$. The foregoing representation of the coefficients is continuous $:PC^\infty(\ol{\Omega}) \to C^\infty (\ol{\Omega})$ with both spaces equipped by the $C^\infty$-topology; a Cauchy sequence in $PC^\infty$ corresponds to a family of Cauchy sequences for the coefficients that are, therefore, convergent in $C^\infty(\ol{\Omega})$; then their limit satisfies the estimate of definition \ref{complex-profile-defn} and it defines a function which is the limit of the Cauchy sequence in $PC^\infty$ we started from. 

Through straightforward calculation we get
\begin{subequations}
\label{prof_der}
\begin{align}
\label{zeta_der}
&\partial_{z_\mu} \mathcal{U}(y,z) = \sum_{(g,\gamma)} \frac{i}{2}(g_\mu + \gamma_\mu) \widehat{U}(y,g,\gamma) e^{i\Psi(g,\gamma;z)},\\
\label{zetabar_der}
&\partial_{\ol{z}_\mu} \mathcal{U}(y,z) = \sum_{(g,\gamma)} \frac{i}{2}(g_\mu - \gamma_\mu) \widehat{U}(y,g,\gamma) e^{i\Psi(g,\gamma;z)},
\end{align}
\end{subequations}
hence, $PC^\infty$ is closed for constant coefficients partial differential operators. In particular, $\mathcal{U}(y,z)$ is holomorphic in $z$ iff $\widehat{U}(y,g,\gamma)=0$ when $g \not= \gamma$. In addition, $PC^\infty$ yields an algebra with respect to pointwise tensor multiplication as envisaged in remark \ref{new:tensor-algebra}. In order to see this, one can write the product of any $\mathcal{U}_1, \mathcal{U}_2 \in PC^\infty$,
\begin{align*}
\mathcal{U}_1(y,z) \mathcal{U}_2(y,z) &= \sum_{(g_1,\gamma_1)} \sum_{(g_2,\gamma_2)} \widehat{U}_1(y,g_1,\gamma_1)\widehat{U}_2(y,g_2,\gamma_2) e^{i\Psi(g_1+g_2,\gamma_1+\gamma_2;z)}\\
&=\sum_{(g,\gamma)} \Big(\sum_{g',\gamma' \leq \gamma} \widehat{U}_1(t,x,g-g',\gamma-\gamma')\widehat{U}_2(t,x,g',\gamma')\Big) e^{i\Psi(g,\gamma;z)},
\end{align*}
where $g = g_1 + g_2$, $\gamma = \gamma_1 + \gamma_2$, $g'=g_2$ and $\gamma'=\gamma_2$ and one can see that the series between brackets converges to $\widehat{U}_3(y,g,\gamma)$ in $C^\infty$; we see that $(g,\gamma)\in \Sigma$ since $|g_\mu| \leq |g_{1,\mu}| + |g_{2,\mu}| \leq \gamma_{1,\mu} + \gamma_{2,\mu}=\gamma_\mu$ and $\widehat{U}_3(y,g,\gamma)$ satisfies the requirements of definition \ref{complex-profile-defn}. 

\begin{proposition}
\label{subspace-profile}
The algebra $PC^\infty(\ol{\Omega})$ constitutes a $C^\infty(\ol{\Omega})$-module closed with respect to nonlinear differential operators with polynomial nonlinearity and coefficients in $C^\infty(\ol{\Omega})$.
\end{proposition}

\begin{proof}
We have already shown that $PC^\infty$ is an algebra closed with respect to constant coefficient partial differential operators. For any $f \in C^\infty(\ol{\Omega})$ and $\mathcal{U} \in PC^\infty$, the product $f(y)\mathcal{U}(y,z)$ is again of the form \eref{complex-profile-series} with coefficients $f(y)\widehat{U}(y,g,\gamma)$ satisfying the requirements of definition \ref{complex-profile-defn}, hence, $f\mathcal{U}\in PC^\infty$. This means that $PC^\infty$ is a $C^\infty(\ol{\Omega})$-module, thus, it is closed with respect to any partial differential operators with coefficients in $C^\infty(\ol{\Omega})$ and, being an algebra, we can also account for polynomial nonlinearities.
\end{proof}

Let us now introduce the subspace of oscillatory profiles which is used in the formulation of the ansatz \eref{new:ansatz}.

\begin{definition}
\label{new:oscillatory}
The space $PC_\osc^\infty(\ol{\Omega})$ is subspace of profiles $\mathcal{U} \in PC^\infty(\ol{\Omega})$ with coefficients $\widehat{U}(t,x,g,\gamma)=0$ when $g=0$.
\end{definition}

The subspace $PC^\infty_\osc$ is closed in the $C^\infty$-topology discussed above, but it is not closed for the pointwise tensor multiplication as rectification occurs \cite{Rect}. According to definition \ref{new:oscillatory} the spectrum of $\mathcal{U} \in PC^\infty_\osc$ is contained in $\Sigma_\osc$ defined in \eref{new:spectrum-osc}.

\section{Solution in the Space of Formal Oscillatory Series}
\label{local-coherence}

According to the formal expansion of section \ref{sec:formal-expansion}, the unknown functions should be determined so that
\begin{gather}
\label{eq1-nonlinear}
P(t,x,\partial_z,\partial_{\ol{z}}) \mathcal{U}^{(0)}(t,x,\phi/\ep ) = O(\ep^{\frac{3}{2}}),\\
\label{eq2-nonlinear}
P(t,x,\partial_z,\partial_{\ol{z}}) \mathcal{U}^{(1)}(t,x,\phi/\ep ) + N(\mathcal{U}^{(0)})(t,x,\phi/\ep) = O(\ep^{\frac{1}{2}}),
\end{gather}
with the operators defined in section \ref{sec:formal-expansion}. 

If equations \eref{eq1-nonlinear} and \eref{eq2-nonlinear} are uniformly satisfied in $\ol{\Omega}$ and the remainders in equation \eref{new:formal-expansion} are continuous in the compact domain $\ol{\Omega}$ and bounded for $\ep \in (0,\ep_0]$, we can conclude that $L(t,x,v^\ep,\partial v^\ep) = O(\ep^{\frac{3}{2}})$ in the quasilinear case $(p=1)$, with the fully nonlinear estimate $L(t,x,v^\ep,\partial v^\ep) = O(\ep^{\frac{1}{2}})$ for the semilinear case ($p=0$).

First, we consider the operator $P(t,x,\partial_{z},\partial_{\ol{z}})$ acting on formal oscillatory series,
\begin{equation*}
\mathcal{U}(t,x,z) = \sum_{(g,\gamma)\in \Sigma_\osc} \widehat{U}(t,x,g,\gamma) e^{i \Psi( g , \gamma ; z)}.
\end{equation*}
By using equations \eref{prof_der} we readily get
\begin{equation*}
P\mathcal{U}(t,x,z) = \sum_{(g,\gamma)\in\Sigma_\osc} \sigma_{L_0}\big(t,x,d\Psi(g,\gamma;\phi)\big) \widehat{U}(t,x,g,\gamma) e^{i\Psi(g,\gamma;z)}.
\end{equation*}
where $\sigma_{L_0}$ is the principal symbol of the linearized operator. Hence, $P$ amounts to a Fourier multiplier with coefficients $\sigma_{L_0}\big(t,x,d\Psi(g,\gamma;\phi)\big)$.
Equation \eref{eq1-nonlinear} is \emph{formally} satisfied if each term of the series evaluated at $z=\phi/\ep$ is $O(\ep^{\frac{3}{2}})$ uniformly in $\ol{\Omega}$. 

Here is where the coherence of complex phases comes into play: we have to fulfill an infinite set of complex eikonal equations by using combinations $\Psi(g,\gamma;\phi)$ of a finite number of complex phases. In several spatial dimensions this is a very strong condition, but in a single space dimension and with the spectrum of profiles defined in section \ref{extension}, the complex phases determined in section \ref{subsec:phases} are always coherent as shown in proposition \ref{new:coherence1}. 

First, let us consider two formal Fourier multipliers equivalent to \eref{new:E} and \eref{new:Q} modulo the equivalence relation \eref{FM-equivalence} with $k=2$ and $k=0$, respectively; we still denote such operators $\mathds{E}$ and $Q$. We want to prove that $\mathds{E}$ amounts to the projector onto the approximate kernel of $P$ and $Q$ amounts to approximate partial inverse of $P$ in the space of formal series. 

As in the linear theory, the key argument is based on the following estimates proved in Part I.

\begin{lemma}
\label{maslov-estimate1}
Let $k\geq 0$ be an integer, $\phi =\varphi + i \chi \in C^\infty (\ol{\Omega};\C)$, $f \in C^\infty(\ol{\Omega})$ and let $S \subset \ol{\Omega}$ be any (non-empty) set such that $\ol{S} \cap \{\chi(t,x)=0\} = \emptyset$. Then,
\begin{equation*}
\ep^{-k}|f e^{i\phi/\ep}| \leq C_{k}, \quad \text{for every $(t,x) \in \ol{S}$, $\ep \in \R_+$},
\end{equation*}
where $C_k = k^k e^{-k} \sup_{(t,x)\in \ol{S}} |f(t,x) /\chi^{k}(t,x)|$.
\end{lemma}

\begin{lemma}
\label{maslov-estimate2-small}
Let $\phi$ and $f$ be as in lemma \ref{maslov-estimate1} and let $R = \{\mathrm{Im}\phi = \chi=0\}$ be a one-dimensional submanifold admitting coordinates $(t,s) = \kappa(t,x)$ on a neighbourhood $\mathcal{O}\subset \ol{\Omega}$ as discussed in section \ref{subsec:phases}. We assume $\chi(t,s) \geq c |s|^q$ in $[0,T] \times \mathcal{I}_s$, with $c>0$ and $q>0$ an even integer. For every $k \in \N$, $t \in [0,T]$ and $s_1,s_2 \in \mathcal{I}_s$ with $s_1<0<s_2$, there are constants $C_k$ such that
\begin{equation*}
\big| \big(f(t,s) - \sum_{n < k} c_n(t) s^n \big)e^{i\phi(t,s)/\ep}\big| \leq \ep^{\frac{k}{q}}\sup_{s_1\leq s \leq s_2} |\partial^n_s f(t,s)|  C_k,
\end{equation*}
uniformly $s \in [s_1,s_2]$, $\ep \in \R_+$, and this estimate can be made uniform in $[0,T]\times [s_1,s_2]$; here, $f(t,s) = f\circ\kappa^{-1}(t,s)$ and analogously for the other functions, whereas $c_n(t) = \partial^n_s f(t,0)/n!$.
\end{lemma}

\begin{remark}
\label{new:applic}
It is worth noting that the complex phases $\phi_\mu$ obtained in section \ref{subsec:phases} are such that $\chi_\mu(t,s) \geq c |s|^2$, then we can apply lemma \ref{maslov-estimate2-small} with $q=2$.
\end{remark}

Let us recall that a Fourier multiplier applied to formal series is $O(\ep^h)$, $h \in \R$, when all the Fourier components of its image evaluated at $z=\phi/\ep$ are $O(\ep^h)$ uniformly in $\ol{\Omega}$, $\ep \in (0,\ep_0]$ where $0<\ep_0<1$. 

\begin{proposition}
\label{PQE}
Any formal Fourier multipliers $\mathds{E}$ and $Q$ equivalent to \eref{new:E} and \eref{new:Q} defined in section \ref{subsec:profiles} satisfy
\begin{align*}
\mathds{E}^2-\mathds{E}=O(\ep^{\frac{3}{2}}),\qquad\qquad\qquad\  &(\text{$\mathds{E}$ is an approximate projector}), \\ 
P \mathds{E} = \mathds{E} P = O(\ep^{\frac{3}{2}}),\qquad\qquad\quad &(\text{$\mathds{E}$ projects into the approximate kernel of $P$}),\\
PQ = QP =I - \mathds{E} + O(\ep^{\frac{1}{2}}),\quad\  &(\text{$Q$ is an approximate partial inverse of $P$}).
\end{align*}
\end{proposition}

\begin{proof}
Let $\omega_{\mu,\ell}$ be functions defined in section \ref{subsec:profiles}, cf., also the proof of proposition 2.5 of Part I. The Fourier components of $\mathds{E}^2-\mathds{E}$ are estimated by
\begin{equation*}
\big|(\pi^2 - \pi)\widehat{U}e^{i\Psi/\ep}\big| \leq \big|(1-\sum \omega_{\mu,\ell})(\pi^2 - \pi)\widehat{U}e^{i\Psi/\ep}\big| + \sum \big|\omega_{\mu,\ell}(\pi^2 - \pi)\widehat{U}e^{i\Psi/\ep}\big|,
\end{equation*}
the sum being on both $\mu$ and $\ell$. In the first term $\mathrm{Im}\Psi$ never vanishes in the support of $(1-\sum_\ell \omega_{\mu,\ell})$; thus, we can apply lemma \ref{maslov-estimate1} and we get the estimate
\begin{subequations}
\label{estimates}
\begin{multline}
\label{estimateP1}
\big|(1-\sum \omega_{\mu,\ell})(\pi^2 - \pi)\widehat{U}e^{i\Psi/\ep}\big| \\
\leq \ep^k \sup_{\ol{S}} \big[|(1-\sum \omega_{\mu,\ell})(\pi^2 - \pi) \widehat{U}|/(\mathrm{Im}\Psi)^k\big] C_k,
\end{multline}
for every $k\in\N$ and uniformly in $\ol{\Omega}$; here, $S = \ol{\Omega} \cap \mathrm{supp}(1-\sum \omega_{\mu,\ell})$. As for the remaining terms, we have to distinguish each case, according to the definition of $\pi$. If $(g,\gamma) \in \Sigma_\mu$, we have
\begin{equation*}
\pi^2 - \pi = \pi_{l(\mu)}^2 - \pi_{l(\mu)} + O(|s|^3) = O(|s|^3);
\end{equation*}
we can apply lemma \ref{maslov-estimate2-small}, with $k=3$ and $q=2$, and get
\begin{equation}
\label{estimateP2}
\big|\omega_{\mu,\ell}(\pi^2 - \pi)\widehat{U}e^{i\Psi/\ep}\big| \leq \ep^{\frac{3}{2}} \sup \big|\partial_s^{3} \big(\omega_{\mu,\ell}(\pi^2 - \pi)\widehat{U}\big)\big|C,
\end{equation}
in $\ol{\Omega}$. If $(g,\gamma) \in \Sigma_\nu$ with $\nu \not = \mu$ or $(g,\gamma) \not\in \Sigma_1 \cup \cdots \cup \Sigma_m$, we have $\mathrm{Im} \Psi>0$ in $\mathrm{supp}(\omega_{\mu,\ell})$ and we can apply again lemma \ref{maslov-estimate1} with the result that
\begin{equation}
\label{estimateP3}
\big|\omega_{\mu,\ell}(\pi^2 - \pi)\widehat{U}e^{i\Psi/\ep}\big| \leq \ep^k \sup_{\ol{S}} \big[|\omega_{\mu,\ell}(\pi^2 - \pi) \widehat{U}|/(\mathrm{Im}\Psi)^k\big] C_k,
\end{equation}
\end{subequations}
for every $k\in\N$, uniformly in $\ol{\Omega}$ and, here, $S = \ol{\Omega} \cap \mathrm{supp}(\omega_{\mu,\ell})$.

Analogously we estimate the Fourier components of $P\mathds{E}$ on noting that when $(g,\gamma)\in \Sigma_\mu$ we have 
\begin{equation*}
-i\sigma_{L_0}(t,x,d\Psi) \pi(t,x,g,\gamma) = V_{\mu}(t,x)\Psi(g,\gamma;\phi) \pi_{l(\mu)} (t,x) = O(|s|^3),
\end{equation*}
in view of proposition \ref{new:riccati} and the same holds for $\pi (-i\sigma_{L_0})$. 

The argument is repeated also for the last assertion, on noting that, 
\begin{align*}
\sigma_{L_0}(t,x,d\Psi) \mathscr{Q}^\phi(t,x,g,\gamma) &= \sum_{l'\not = l(\mu)} \frac{-i \sigma_{L_0}(t,x,d\langle g,\varphi \rangle) \pi_{l'}(t,x)}{\langle g, V_{l'}(t,x)\varphi(t,x) \rangle}\Bigg|_{R_\mu} + O(|s|)\\
& = \sum_{l'\not = l(\mu)} \pi_{l'}(t,x)_{|_{R_\mu}} + O(|s|)\\
& = (I - \pi(t,x,g,\gamma)) + O(|s|),
\end{align*}
when $(g,\gamma) \in \Sigma_\mu$. The same also holds for $\mathscr{Q}^\phi \sigma_{L_0}$.
\end{proof}

We can now go back to equations \eref{eq1-nonlinear} and \eref{eq2-nonlinear}. Let us simplify the notation by writing $\mathcal{U}_j$ instead of $\mathcal{U}^{(j)}$ for $j=0,1$. Equation \eref{eq1-nonlinear} reads
\begin{equation*}
P\mathcal{U}_0  = P \mathds{E} \mathcal{U}_0 + P(I-\mathds{E}) \mathcal{U}_0 = O(\ep^{\frac{3}{2}}),
\end{equation*}
which is equivalent to
\begin{equation}
\label{amp1-nonlinear}
P(I-\mathds{E})\mathcal{U}_0 = O(\ep^{\frac{3}{2}}).
\end{equation}
Upon multiplication by the partial inverse $Q$, this yields $(I-\mathds{E})\mathcal{U}_0 = O(\ep^{\frac{1}{2}})$, but this is not sharp as we shall see in corollary \ref{amplitudes-formal} below. Moreover, on multiplying equation \eref{eq2-nonlinear} by $\mathds{E}$ we find the necessary condition
\begin{equation}
\label{amp2-nonlinear}
\mathds{E} N = O(\ep^{\frac{1}{2}}),
\end{equation}
whereas, on multiplying by $I-\mathds{E}$ and on setting $\mathds{E}\mathcal{U}_1=0$, we get
\begin{equation}
\label{amp3-nonlinear}
(I-\mathds{E})P(I-\mathds{E}) \mathcal{U}_1 + (I-\mathds{E})N = O(\ep^{\frac{1}{2}}).
\end{equation}

\begin{corollary}
\label{amplitudes-formal}
Let $\mathcal{U}_0$, $\mathcal{U}_1$, $\mathcal{V} =N(\mathcal{U}_0)$ and $\mathcal{W}$ be formal oscillatory series with $\mathds{E} \mathcal{V} = O(\ep^{\frac{1}{2}})$ and $(I- \mathds{E})\mathcal{W}=O(\ep^{\frac{1}{2}})$. Then, 
\begin{itemize}
\item[a)] equation \eref{eq1-nonlinear} is equivalent to $(I-\mathds{E})\mathcal{U}_0=O(\ep^{\frac{3}{2}})$,
\item[b)] the formal solution of equation \eref{eq2-nonlinear} is $\mathcal{U}_1 = -Q\mathcal{V} + \mathcal{W}$.
\end{itemize}
\end{corollary}

\begin{proof}
a) We have just shown that equation \eref{eq1-nonlinear} is formally equivalent to \eref{amp1-nonlinear} which is implied by $(I-\mathds{E})\mathcal{U}_0 = O(\ep^{\frac{3}{2}})$. We have to show that this is a necessary condition too. 
It is enough to look at the Fourier coefficients of $P(I-\mathds{E})\mathcal{U}_0$ with $(g,\gamma) \in \Sigma_\mu$ in the support of $\omega_{\mu,\ell}$ that are
\begin{equation*}
\sum_{l \not = l(\mu)} (V_l(t,x)\Psi) \pi_l(t,x) \widehat{U}(t,x,g,\gamma) + O(|s|^3). 
\end{equation*}
In view of proposition \ref{new:coherence2}, $|V_l(t,x)\Psi|\geq C>0$ in the support of $\omega_{\mu,\ell}$, at least after shrinking the open sets $\mathcal{O}_{\mu,\ell}$ in the $s$ direction, thus, equation \eref{amp1-nonlinear} entails
\begin{equation*}
\sum_{l\not=l(\mu)} \pi_l(t,x) \widehat{U}(t,x,g,\gamma) e^{i\Psi/\ep} = O(\ep^{\frac{3}{2}}),
\end{equation*}
and on the left-hand side there are just the coefficients of $(I-\mathds{E})\mathcal{U}_0$. 

b) Let us note that, $PQ\mathcal{V} = (I-\mathds{E})\mathcal{V} + O(\ep^{\frac{1}{2}}) = \mathcal{V} + O(\ep^{\frac{1}{2}})$, whereas $P\mathcal{W} = P(I-\mathds{E})\mathcal{W} +O(\ep^{\frac{3}{2}})$, with $P(I-\mathds{E})\mathcal{W}=O(\ep^{\frac{1}{2}})$ by hypothesis.
\end{proof}

\section{Solution in $PC^\infty_\osc$ and the Transport Equation}
\label{smooth-solvability}

In this section, we shall deal with the convergence of the formal series addressed in section \ref{local-coherence}. 

By definition \ref{complex-profile-defn}, we know that $| \widehat{U}(t,x,g,\gamma)| \leq C_k |(g,\gamma)|^{-k}$ for every $k \in \N$ uniformly in $\ol{\Omega}$, and the same holds for any derivative $\partial_t^n\partial_{x}^m \widehat{U}(t,x,g,\gamma)$. On the other hand, $\sigma_{L_0}(t,x,d\Psi(g,\gamma;\phi))$ is homogeneous of degree $1$ in $(g,\gamma)$ so that $\sigma_{L_0}\widehat{U} = O(|(g,\gamma)|^{-k})$ for every $k$, hence, the series of such coefficients is uniformly convergent in $\ol{\Omega} \times \T_c^m$. Moreover, the same holds for the coefficients of $\partial^n_t \partial_{x}^m P\mathcal{U}(t,x,z)$, hence, $P$ is continuous $:PC^\infty \to PC^\infty$ and restricts to a continuous operator $:PC^\infty_\osc \to PC^\infty_\osc$ still denoted $P$. This is easy because $P$ is a first-order differential operator, thus, it is clearly well defined on $C^\infty$ functions. On the other hand, the case of $\mathds{E}$ and $Q$ requires some efforts.

\begin{proposition}
\label{smooth-EQ}
The formal operators \eref{new:E} and \eref{new:Q} are continuous Fourier multipliers $:PC_\osc^\infty \to PC_\osc^\infty$ for which the formal estimates proved in proposition \ref{PQE} and corollary \ref{amplitudes-formal} are rigorously valid in $PC_\osc^\infty$.
\end{proposition}

\begin{proof}
If $\mathcal{U}\in PC^\infty$ and $\alpha\in\N^2$, $\big|\partial_{t,x}^\alpha \big(\pi(t,x,g,\gamma) \widehat{U}(t,x,g,\gamma)\big)\big| \leq C |(g,\gamma)|^{-k}$ for every $k \in \N$ since the coefficients $\pi(t,x,g,\gamma)$ and all their derivatives are bounded for $(t,x)\in\ol{\Omega}$ and $(g,\gamma)\in\Sigma$. Therefore, $\mathds{E}$ is continuous $:PC_\osc^\infty \to PC_\osc^\infty$.

As for \eref{new:Q}, if $(g,\gamma) \in \Sigma_\mu$, for every $k \in \N$,
\begin{equation*}
|\mathscr{Q}^\phi (t,x,g,\gamma)| \leq \sum_{l'\not=l} \left|\frac{\pi_{l'}(t,x)}{V_{l'}(t,x)\langle g,\varphi(t,x) \rangle} \right|, 
\end{equation*}
in a compact near $R_\mu$, hence, $|\mathscr{Q}^\phi| \leq C |(g,\gamma)|^{-1}$ in virtue of proposition \ref{new:coherence2}. We need to control the derivative $\partial_t\mathscr{Q}^\phi$ only. We have
\begin{equation*}
\partial_t \mathscr{Q}^\phi(t,x,g,\gamma) = \sum_{l'\not =l}\Bigg[ \frac{\partial_t \pi_{l'}(t,x)}{V_{l'}(t,x)\langle g,\varphi \rangle} -\frac{\partial_t (V_{l'}(t,x)\langle  g,\varphi\rangle)}{(V_{l'}(t,x)\langle  g,\varphi\rangle)^2} \pi_{l'}(t,x) \Bigg],\ (t,x) \in R_\mu,
\end{equation*}
thus, $|\partial_t \mathscr{Q}^\phi| \leq C |(g,\gamma)|^{-1}$, again because of proposition \ref{new:coherence2} and this extends to higher orders, i.e., $|\partial_t^n \mathscr{Q}^\phi|\leq C_n |(g,\gamma)|^{-1}$. Then, $Q$ is continuous $:PC^\infty_\osc \to PC^\infty_\osc$.

As for the formal estimate in proposition \ref{PQE}, it is enough noting that the coefficients in the right-hand sides of estimates \eref{estimates} and analogous are all as rapidly decreasing as $\widehat{U}$, hence, the corresponding series are absolutely convergent. The same holds also for corollary \ref{amplitudes-formal}. 
\end{proof}

Now we prove that the estimate of proposition \ref{PQE} and corollary \ref{amplitudes-formal} hold for any pair Fourier multiplier $\mathds{E}$ and $Q$ that are equivalent to \eref{new:E} and \eref{new:Q}, respectively. More generally, the next result shows what happens when one varies a Fourier multiplier $:PC_\osc^\infty \to PC_\osc^\infty$ within a fixed equivalence class.

\begin{proposition}
\label{modulo}
If $A_1$ and $A_2$ are Fourier multipliers $:PC_\osc^\infty \to PC_\osc^\infty$ equivalent modulo the relation \eref{FM-equivalence}, then $A_1-A_2 =O(\ep^{\frac{k+1}{2}})$ uniformly in $\ol{\Omega}$.
\end{proposition}

\begin{proof}
Let $\mathscr{A}_j(t,x,g,\gamma)$ be the coefficients of $A_j$. We apply our usual argument of propositions \ref{PQE} and \ref{smooth-EQ}, and we write 
\begin{equation*}
\big| (\mathscr{A}_1-\mathscr{A}_2) \widehat{U} e^{i\Psi/\ep}\big| \leq
\big|(1-\sum \omega_{\mu,\ell}) (\mathscr{A}_1-\mathscr{A}_2) \widehat{U} e^{i\Psi/\ep}\big| + \sum  \big| \omega_{\mu,\ell}(\mathscr{A}_1-\mathscr{A}_2) \widehat{U} e^{i\Psi/\ep}\big|,
\end{equation*}
with the sum being over both $\mu$ and $\ell$. For the first term we apply lemma \ref{maslov-estimate1} and show that it is $O(\ep^\infty)$. The remaining terms are estimated again by lemma \ref{maslov-estimate1} if $(g,\gamma) \in \Sigma_\nu$ with $\nu \not = \mu$ or $(g,\gamma) \not\in \Sigma_1 \cup\cdots\cup\Sigma_m$ and by lemma \ref{maslov-estimate2-small} if $(g,\gamma) \in \Sigma_\mu$ which, in particular, gives the $O(\ep^{\frac{k+1}{q}})$ estimate with $q=2$. Moreover, as in proposition \ref{smooth-EQ}, the series are convergent.
\end{proof}

On the basis of proposition \ref{amplitudes-formal}, strengthened in proposition \ref{smooth-EQ}, we have to find $\mathcal{U}_0$ such that $(I-\mathds{E})\mathcal{U}_0 = O(\ep^{\frac{3}{2}})$ and $\mathds{E} N(\mathcal{U}_0)=O(\ep^{\frac{1}{2}})$. We start with the analogous of proposition 5.1 of Part I. We shall denote by $\mathds{E}_0$ the restriction of $\mathds{E}$ to a Fourier multiplier $:PC_\osc^\infty(R;\C^N) \to PC^\infty_\osc(R;\C^N)$.

\begin{proposition}
\label{furbata2}
Let $\mathcal{U} \in PC_\osc^\infty(R;\C^N)$ be such that $(I-\mathds{E}_0) \mathcal{U} = 0$ and let $\underline{\mathcal{U}} \in PC^\infty_\osc(\ol{\Omega};\C^N)$ be a smooth extension of $\mathcal{U}$ to a compact neighbourhood of $R$. Then, on defining $\mathcal{U}_0(t,x,z) = \mathds{E} \underline{\mathcal{U}}(t,x,z)$, $\mathds{E}$ being any operator that belongs to the class of \eref{new:E}, we have
\begin{equation*}
(I-\mathds{E})\mathcal{U}_0(t,x,\phi/\ep) = O(\ep^{\frac{3}{2}}),\quad \text{and}\quad \mathcal{U}_0(t,x,\phi/\ep) - \underline{\mathcal{U}}(t,x,\phi/\ep) = O(\ep^{\frac{1}{2}}),
\end{equation*}
the estimates being uniform in $\ol{\Omega}$.
\end{proposition}

\begin{proof}
The first claim is a consequence of the first estimate in proposition \ref{PQE} strengthened in proposition \ref{smooth-EQ}; the second one comes from $(I-\mathds{E}_0)\mathcal{U}=0$ upon Taylor expanding the Fourier coefficients of $\mathds{E}$ in a neighbourhood of $R$ in the expression for $\mathcal{U}_0$.
\end{proof}

Let us, now, consider the solvability condition $\mathds{E} N(\mathcal{U}_0)(t,x,\phi/\ep)=O(\ep^{\frac{1}{2}})$ which is well-posed since the image of $N$ belongs to the domain of the approximate projector $\mathds{E}$ in view of assumption \ref{rect}. 

\begin{proposition}
\label{te1}
If $\mathcal{U} \in PC_\osc^\infty(R;\C^N)$ solves the system \eref{te} and $\underline{\mathcal{U}}$ is defined as in proposition \ref{furbata2}, then $\mathcal{U}_0 = \mathds{E} \underline{\mathcal{U}}$ fulfills the estimates $(I-\mathds{E}) \mathcal{U}_0(t,x,\phi/\ep) = O(\ep^{\frac{3}{2}})$ and $\mathds{E} N(\mathcal{U}_0)(t,x,\phi/\ep) = O(\ep^{\frac{1}{2}})$.
\end{proposition}

\begin{proof}
By proposition \ref{furbata2} we already know that $(I-\mathds{E})\mathcal{U}_0 = O(\ep^{\frac{3}{2}})$, for every solution $\mathcal{U}$. As for the remaining estimate, we look at the Fourier coefficients of $\mathds{E}N(\mathcal{U}_0)$ and note that, when they are evaluated on $R$, they amount exactly to the coefficients of
\begin{equation*}
\mathds{E}_0\big[L_0 + B(t,x,\mathcal{U})\partial_\theta + C(t,x,\mathcal{U})\big] \mathcal{U},
\end{equation*}
and, thus, when $\mathcal{U}$ is a solution of \eref{te}, the coefficients of $\mathds{E}N(\mathcal{U}_0)$ are $O(|s|)$ near $R$. The nonlinear terms in $\mathscr{B}_0(\mathcal{U}_0)$, in particular, amounts to 
\begin{equation*}
B(t,x,\mathcal{U}_0)\partial_{\theta} \mathcal{U}_0 + \sum_\mu (\partial_uA \mathcal{U}_0 +\partial_{\ol{u}}A \ol{\mathcal{U}}_0)\partial_x \chi_\mu \partial_{y_\mu}\mathcal{U}_0,
\end{equation*}
and the coefficients of the differential operator $\sum_\mu d_x\chi_\mu \partial_{y_\mu}$ are $-\sum_\mu d_x\chi_\mu\gamma_\mu=d_x\mathrm{Im}\Psi=O(|s|)$ near $R$. Then the estimate is proved via usual arguments.
\end{proof}

Therefore, one can construct the extended profiles by solving the Cauchy problem for the system \eref{te}; specifically, one has the following result.

\begin{proposition}
\label{profiles-construction}
Let us assume that the Cauchy problem \eref{te} is well-posed in $PC^\infty_\osc(R;\C^N)$ and define the profiles as in section \ref{subsec:profiles}. Then, equations \eref{eq1-nonlinear} and \eref{eq2-nonlinear} are satisfied and $\mathcal{U}_{0|t=0}(x,\phi_{|t=0}/\ep) = \mathcal{H}(x, \phi_{|t=0}/\ep) + O(\ep^{\frac{1}{2}})$.
\end{proposition}

\begin{proof}
Let $\mathcal{U} \in PC^\infty_\osc(R;\C^N)$ be the solution of the Cauchy problem \eref{te}. In virtue of propositions \ref{smooth-EQ} and \ref{te1}, $\mathcal{U}_0 =\mathds{E}\underline{\mathcal{U}}$ is such that $P\mathcal{U}_0 = O(\ep^{\frac{3}{2}})$ and $\mathds{E} N(\mathcal{U}_0) = O(\ep^{\frac{1}{2}})$, uniformly in $\ol{\Omega}$. Corollary \ref{amplitudes-formal} and proposition \ref{smooth-EQ} ensure that the profile $\mathcal{U}_1 = - Q N(\mathcal{U}_0)$ satisfies equation \eref{eq2-nonlinear} in the compact domain $\ol{\Omega}$. 

Finally, we have $\mathcal{U}_{0|t=0}(x,\phi_{|t=0}/\ep) = \mathds{E}_{|t=0} \mathcal{H}(x,\phi_{|t=0}/\ep) + O(\ep^{\frac{1}{2}})$, and $(I-\mathds{E}_{|t=0})\mathcal{H}=0$ in virtue of condition \ref{pol}. 
\end{proof}

\section{Existence of Profiles}
\label{existence}

In this section we shall address the well-posedness of the Cauchy problem for the nonlinear transport equation \eref{te}. The strategy is to reduce the system \eref{te} to the transport equation for periodic profiles and apply the classical results of standard nonlinear geometric optics.

As in section \ref{extension}, for every compact manifold $\ol{\Omega}$, we can argue that an element of $PC_\osc^\infty(\ol{\Omega})$, regarded as a function $\mathcal{U}(t,x,\theta,r)$, $\theta = \mathrm{Re}(z)$ and $r =\mathrm{Im}(z)$, can be extended to $\mathcal{U}(t,x,\theta,w)$ with $w \in \C^m$ and $\mathrm{Re}(w_\mu) = r_\mu \geq 0$. On evaluating the extended function for $w = -i\theta'$, $\theta'\in\T^m$, we obtain the periodic profile 
\begin{equation}
\label{rot1}
U(t,x,\theta,\theta') = \mathcal{U}(t,x,\theta,-i\theta')\in C^\infty(\ol{\Omega}\times \T^{m} \times \T^m), 
\end{equation}
which has been used in section \ref{extension} in order to represent the coefficients $\widehat{U}(t,x,g,\gamma)$. The argument can clearly be inverted: any periodic profile $U(t,x,\theta,\theta') \in C^{\infty}(\ol{\Omega}\times \T^m\times \T^m)$ such that $\widehat{U}(t,x,g,h)=0$ when $(g,h) \not\in \Sigma_\osc$ (the spectrum of oscillatory profiles) can be extended to a smooth function $U(t,x,\theta,w)$ with $w \in \C^m$ and $\mathrm{Im}(w_\mu) = r_\mu \geq 0$; then, we can set
\begin{equation} 
\label{rot2}
\mathcal{U}(t,x,z) = U(t,x,\theta,ir), \qquad z=\theta +ir,
\end{equation}
and this belongs to $PC_\osc^\infty$. The mappings \eref{rot1} and \eref{rot2} are one the inverse of the other, therefore, we have an injective linear map 
\begin{equation}
\label{injection}
j : \mathcal{U}(t,x,z) \mapsto U(t,x,\theta,\theta'), 
\end{equation}
which is continuous in the $C^\infty$-topology and the image $jPC_\osc^\infty$ is a closed subspace of $C^\infty(\ol{\Omega}\times \T^m\times \T^m)$. The projector $\mathds{P}$ into the image $jPC_\osc^\infty$ of the injection amounts to the Fourier multiplier with coefficients
\begin{equation}
\label{projector-P}
\mathscr{P}(g,h) = \left\{
\begin{aligned}
&I, \qquad &\text{when $(g,h) = (g,\gamma) \in \Sigma_\osc$},\\
&0, &\text{otherwise}.
\end{aligned}\right.
\end{equation}

\begin{proposition}
\label{reduction}
The complex profile $\mathcal{U} \in PC_\osc^\infty(R;\C^N)$ satisfies the system \eref{te} iff $U = j\mathcal{U} \in C^\infty(R \times \T^m\times \T^m;\C^N)$ satisfies
\begin{equation}
\label{transport-equation-nonlinear}
\left\{
\begin{aligned}
& (I-\mathds{F})U = 0,\\
& \mathds{F} \big[L_0 + B(t,x,U)\partial_\theta + C(t,x,U) \big] U = 0,\qquad (t,x)\in R,\\
& U_{|t=0}(x,\theta,\theta') = H(x,\theta,\theta'), \qquad x \in R\cap X^o,
\end{aligned}\right.
\end{equation}
where $\mathds{F} = \mathds{P}\mathds{E}_0=\mathds{E}_0\mathds{P}$ and $H=j\mathcal{H}$.
\end{proposition}

\begin{proof}
Let $N_1(\mathcal{U}) = L_0\mathcal{U} + B(t,x,\mathcal{U}) \partial_\theta \mathcal{U} + C(t,x,\mathcal{U})\mathcal{U}$ and $N_2(U)  = \mathds{P} [L_0 + B(t,x,U)\partial_\theta  + C(t,x,U)] U$. First, we note that $N_1$ and $N_2$ are conjugated by $j$, that is, $j N_1(\mathcal{U}) = N_2(j\mathcal{U})$ and, when $U$ belongs to the range of $j$, $j^{-1} N_2(U) = N_1(j^{-1}U)$. In addition, $j$, $j^{-1}$ and the projector $\mathds{P}$ commute with $\mathds{E}_0$. Therefore, by acting with $j$ on the system \eref{te} and recalling that $j\mathcal{U} = \mathds{P}j\mathcal{U}$ one finds \eref{transport-equation-nonlinear} with $U=j\mathcal{U}$. Vice versa, by acting on \eref{transport-equation-nonlinear} by $j^{-1}$ and, on recalling that a solution $U$ must belong to $jPC_\osc^\infty$, one finds the system \eref{te}.
\end{proof}

The transport equation \eref{transport-equation-nonlinear} is the restriction to $R$ of the classical transport equation for periodic profiles on $\ol{\Omega}$ studied by Joly, M\'etivier and Rauch \cite{JMR1} with few slight modifications. This gives the existence of the solution. In what follows we give a somewhat concise discussion of the well-posedness for \eref{transport-equation-nonlinear}, based on the work of Joly, M\'etivier and Rauch \cite{JMR1,JMR2} on the existence of periodic profiles.

Let us consider \eref{transport-equation-nonlinear} on each connected component $R_{\mu,\ell}$ of $R$. For $(t,x)\in R_{\mu,\ell}$, the coefficients of $\mathds{F}$ are equal to $\pi_{l(\mu)}(t,x)$ when $(g,\gamma) \in \Sigma_\mu$ and equal to zero otherwise ($\pi(t,x,g,\gamma)$ is supported away from $R_\mu$ if $(g,\gamma) \in \Sigma_\osc \setminus \Sigma_\mu$ and $\mathscr{P}(g,\gamma)=0$ if $(g,\gamma) \not = \Sigma_\osc$); hence, $\mathds{F}$ evaluated on $R_{\mu,\ell}$ is a projector. As a consequence, $(I-\mathds{F})U=0$ implies that $\widehat{U}(t,x,g,\gamma)=0$ when $(t,x) \in R_{\mu,\ell}$, $(g,\gamma) \not \in \Sigma_\mu$ and $(I-\pi_{l(\mu)}(t,x)) \widehat{U}(t,x,g,\gamma)=0$ when $(t,x) \in R_{\mu,\ell}$ and $(g,\gamma) \in \Sigma_\mu$. It follows that the non-zero coefficients of $\mathds{F} L_0 U$ are
\begin{multline*}
\pi_{l(\mu)} (t,x) \big(\partial_t + A_0(t,x)\partial_x \big) \pi_{l(\mu)} (t,x) \widehat{U}(t,x,g,\gamma) \\
= \pi_{l(\mu)}(t,x) V_\mu(t,x) \widehat{U} + \text{lower order terms},
\end{multline*}
where $(t,x)\in R_{\mu,\ell}$ and the smooth vector field $V_\mu(t,x) = \partial_t + \lambda_{l(\mu),0}(t,x)\partial_x$ is tangent to $R_{\mu,\ell}$.

By using the time coordinate $t \in [0,T]$ to parametrize each $R_{\mu,\ell}$, the nonlinear transport equation \eref{transport-equation-nonlinear} splits into a set of initial value problems for $U \in C^\infty([0,T];\T^n)$ of the form
\begin{equation*}
\left\{
\begin{aligned}
& (I - \mathds{F})U(t,\theta) = 0,\\
& \mathds{F}\big[\partial_t + B(t,U)\partial_\theta + C(t,U)\big]U(t,\theta)=0,\\
& U_{|t=0}(\theta) = H(\theta) \in C^\infty(\T^n).
\end{aligned}\right.
\end{equation*}
The only difference, here, with respect to the standard theory of hyperbolic symmetric systems on the torus $\T^n$, \cite{T3}, is the presence of the projector $\mathds{F}$. On the other hand, with respect to the systems considered by Joly, M\'etivier and Rauch \cite{JMR1,JMR2}, the restriction to the reference manifold $R$ has eliminated the spatial degree of freedom.

For this simple case, the space $\mathscr{E}^s(T)$ of periodic profiles \cite{JMR1,JMR2} can be taken to be $C([0,T];H^s(\T^n))$, $H^s(\T^n)$ being the standard Sobolev distributions on $\T^n$ with index $s \in \R$. With the norm
\begin{equation*}
\|U\|_{\mathscr{E}^s(T)} = \sup_{0\leq t \leq T} \|U(t)\|_{H^s(\T^n)},
\end{equation*}
$\mathscr{E}^s(T)$ is a Banach space. In addition, the properties of $H^s$ imply that, for $s >n/2$, $\mathscr{E}^s(T)$ is a Banach algebra embedded in $L^\infty([0,T]\times \T^n)$ in which the composition with smooth functions $F(t,x,U)$ is continuous and bounded, [6-8, 15]. 

The projector $\mathds{F}$ is continuous on $\mathscr{E}^s(T)$ and for every $t\in [0,T]$ the operator $\mathds{F}(t)$ obtained by freezing the time coordinate in the coefficients is a projector in $H^s(\T^n)$ orthogonal with respect to the $L^2(\T^n)$ product. The range of $\mathds{F}$ in $\mathscr{E}^s(T)$ is denoted $\mathscr{N}^s(T)$ and the condition $(I - \mathds{F})U =0$, $U \in \mathscr{E}^s(T)$ is equivalent to $U \in \mathscr{N}^s(T)$.

The well-posedness of such systems can be proved by means of Picard iterates along the usual lines. For sake of completeness we outline the main points without proofs: details can be found in [6-8]. First, one considers the linearized system
\begin{equation}
\label{linearized}
\left\{\begin{aligned}
&U \in \mathscr{N}^s(T),\\
&\mathds{F}\big[\partial_t + B(t,V) \partial_\theta + C(t,V)\big] U =0,\\
&U_{|t=0} = H \in \mathscr{N}^s(0)=H^s(\T^n),
\end{aligned}\right.
\end{equation}
with $V,G \in \mathscr{E}^{s}(T)$ being given.

\begin{theorem}
\label{linearized-wp}
If $V \in \mathscr{E}^s(T)$ and $H \in \mathscr{N}^s(0)$, with $s > n/2 + 1$, then the linearized system \eref{linearized} has a unique solution $U\in\mathscr{N}^s(T)$ with
\begin{equation}
\label{linear-estimate}
\|U(t)\|_{H^{s'}}^2 \leq e^{Ct}\|H\|_{H^{s'}}^2,
\end{equation}
for every $t \in [0,T]$, $s' \leq s$; here $C = 1 + \|\partial_x\cdot A\| + M (\|V\|+\|\partial_\theta V\|)$ with norms taken in $L^\infty([0,T] \times \T^n)$.
\end{theorem}

Then, one construct the solution of \eref{transport-equation-nonlinear} as the limit of a sequence $U_1$, $U_2,\ldots$, $U_\nu,\ldots$ of Picard's iterates: let us set $U_1(t,\theta) = H(\theta) \in \mathscr{E}^s(T)$, with $H \in \mathscr{N}^s(0)$ and $s >n/2+1$, then let $U_\nu \in \mathscr{E}^s(T)$ for $\nu \geq 2$ be the solution of
\begin{equation}
\label{iterates}
\left\{\begin{aligned}
&U_\nu \in \mathscr{N}^s(T),\\
&\mathds{F} \big[\partial_t + B(t,U_{\nu-1})\partial_\theta + C(t,U_{\nu-1})\big]U_\nu =0,\\
&U_{\nu|t=0} = H \in \mathscr{N}^s(0).
\end{aligned}\right.
\end{equation}
Iterative use of theorem \ref{linearized-wp} yields a sequence $U_\nu \in \mathscr{N}^s(T)$, $\nu \geq 1$.

\begin{theorem}
\label{Picard}
Let $H\in \mathscr{N}^s(0)$ and $s>n/2+1$. Then, there exists $t_* \in [0,T]$, such that the sequence $U_\nu$ is bounded in $\mathscr{N}^s(t_*)$ and converges in $\mathscr{N}^{s-1}(t_*)$ to $U$ which is the unique solution of the considered Cauchy problem. Moreover, $t_*$ can be made independent on $s$, so that, if $H\in \mathscr{N}^\infty(0)$, $U \in C^\infty(\ol{\Omega}\times \T^m)$.
\end{theorem}

Let us redefine $T$ to be the maximum between the existence time $t_*$ for the profiles and the old $T$, then, theorem \ref{Picard}, proposition \ref{reduction} and proposition \eref{profiles-construction} imply the existence of complex geometric optics profiles in $\ol{\Omega} \times \T^m_c$. 

\begin{corollary}
\label{solution-profile}
There exists a class of equivalence of complex geometric optics profiles $\mathcal{U}_0$ and $\mathcal{U}_1$ satisfying the conditions stated in proposition \eref{profiles-construction}.
\end{corollary}

\section{Proof of the Main Results}
\label{sec:proof-main}

With respect to the corresponding result in the linear theory, proposition 2.5 of Part I, the proof is shorter as all the detail on the local structure of the wave field around the reference manifold $R$ have been already accounted for in proposition \ref{PQE} and its strengthened version, proposition \ref{smooth-EQ}.

\begin{proof}[Proof of proposition \ref{main}]
First, we address the initial values. The lowest-order nonlinear complex geometric optics solution $v^\ep$ restricted to $t=0$ is such that
\begin{align*}
u_{|t=0}^\ep (x) - v^\ep_{|t=0}(x) &= \ep^{p} \Big[ \sum h_\mu (x) e^{i\psi_\mu(x)/\ep} - \mathcal{H}\big(x,\phi_{|t=0}(x)/\ep\big)\Big] + O(\ep^{p+\frac{1}{2}})\\
&= \ep^p \sum h_\mu(x) \Big[e^{i\psi_\mu(x)/\ep}- e^{i \phi_{\mu|t=0}(x)/\ep}\Big] + O(\ep^{p+\frac{1}{2}}), 
\end{align*}
uniformly in $\ol{X^o}$, with $p=1$ and $p=0$ in the quasilinear and semilinear case, respectively. We can apply to each term in the sum the same argument used in the proof of assertion a) of proposition 2.5 in Part I and we readily obtain the claimed estimate.

As for assertion b), we note that the remainders in equation \eref{new:formal-expansion} are smooth on $\ol{\Omega}$, thus, in particular, they are bounded in $\ol{\Omega}$ and yields a contribution of order $O(\ep^{p+1})$. Then, by a straightforward application of corollary \ref{solution-profile} and proposition \eref{profiles-construction}, we can find a solution of equation \eref{eq1-nonlinear} and \eref{eq2-nonlinear}, with the result that $L(t,x,v^\ep,\partial v^\ep) = O(\ep^{p+\frac{1}{2}})$.
\end{proof}

\section*{Acknowledgments}
This work has been supported partly by the CNISM, at the Department of Physics ``A.~Volta'' of the Pavia University (Italy), under the grant ``Propagazione di onde ad alta frequenza in mezzi dispersivi e disomogenei: dalla teoria dei sistemi dinamici alle applicazioni'' and by the Foundation Blanceflor Boncompagni-Ludovisi at the Max-Planck-Institut f\"ur Plasmaphysik (IPP), Garching bei M\"unchen (Germany). I thank the Theory Division of the IPP and, in particular, G.~V.~Pereverzev and E.~Poli for their kind hospitality and collaboration. I am greatly indebted to M.~Bornatici for his continuous encouragement and tutorship. Especially, I wish to thank C.~Dappiaggi for so many discussions, suggestions and for carefully reading the manuscript.

\end{document}